\documentclass[pra,10pt,twocolumn,superscriptaddress,showpacs]
{revtex4-1}
\usepackage{amsmath,amsthm}
\usepackage{float}
\usepackage{graphicx}
\usepackage{caption}
\usepackage{latexsym}
\usepackage{amssymb}
\usepackage{mathtools}
\usepackage{braket}
\usepackage{dsfont}
\usepackage{comment}
\usepackage[colorlinks=true, citecolor=blue, urlcolor=blue]{hyperref}
\usepackage{amsfonts}

\newcommand{\be}{\begin{equation}}
\newcommand{\ee}{\end{equation}}
\newcommand{\bea}{\begin{eqnarray}}
\newcommand{\eea}{\end{eqnarray}}

\newcommand{\blk}{\color{black}}

\makeatletter
\newtheorem*{rep@theorem}{\rep@title}
\newcommand{\newreptheorem}[2]{%
\newenvironment{rep#1}[1]{%
 \def\rep@title{#2 \ref{##1}}%
 \begin{rep@theorem}}%
 {\end{rep@theorem}}}
\makeatother

\newtheorem{theorem}{Theorem}
\newreptheorem{theorem}{Theorem}

\begin{document}

\title{Wigner's approach enabled detection of multipartite nonlocality using all different bipartitions }

\author{Sumit Nandi}
\email{sumit.nandi@bose.res.in}
\affiliation{S. N. Bose National Centre for Basic Sciences, Block JD, Sector III, Salt Lake, Kolkata 700 106, India}
\author{Debashis Saha}
\email{saha@bose.res.in}
\affiliation{S. N. Bose National Centre for Basic Sciences, Block JD, Sector III, Salt Lake, Kolkata 700 106, India}
\affiliation{School of Physics, IISER Thiruvananthapuram, Kerala, India 695551}

\author{Dipankar Home}
\email{quantumhome80@gmail.com}
\affiliation{Center for Astroparticle Physics and Space Science (CAPSS), Bose Institute, Kolkata 700 091, India}

\author{A. S. Majumdar}
\email{archan@bose.res.in}
\affiliation{S. N. Bose National Centre for Basic Sciences, Block JD, Sector III, Salt Lake, Kolkata 700 106, India}

\begin{abstract}

Distinct from Bell's approach, Wigner had derived a form of local realist (LR) inequality which is quantum mechanically violated for a bipartite maximally entangled state. Subsequently, this approach was generalized to obtain a multipartite LR inequality. However, the violation of such generalised Wigner's inequality does not guarantee nonlocality between all possible different bipartitions of the multipartite system. In the present work, this limitation has been overcome by formulating a further generalisation of Wigner's approach through the derivation of a set of LR inequalities with respect to all different bipartitions of a N-partite system. Quantum mechanical violations of all individual LR inequalities belonging to such a set would rigorously certify multipartite nonlocality by also providing a finer characterisation of the nature of multipartite nonlocality in the following sense. The quantum mechanical violation of any given inequality of our complete set of LR inequalities would enable identification of the corresponding bipartition which exhibits nonlocality. This is in contrast to other multipartite LR inequalities such as the Svetlichny inequality or its generalisation that cannot be used to detect whether there is any particular bipartition which is nonlocally correlated. The efficacy of the scheme developed in this paper is illustrated for the tripartite and quadripartite states.
\end{abstract} 

\maketitle
\section{Introduction}\label{intro}
Discovery of Bell's inequality (BI) \cite{bell64,bell69} has given rise to extensive theoretical and experimental investigations over nearly six decades, thereby establishing nonlocality as a deep-seated fundamental feature of nature, compatible with quantum mechanical predictions (see, for example, the comprehensive review of the relevant theory works by Brunner \emph{et al} \cite{brunner}, and the most conclusive recent tests of BI \cite{nature,prl1,prl2}). Apart from its wide-ranging foundational implications, nonlocality has also been shown to be a valuable resource for accomplishing quantum information processing protocols like ensuring secret key distribution \cite{ekert}, device independent certification of randomness \cite{pironio}, and for efficient implementations of communication complexity protocols \cite{zukowski}. However, most of the studies have so far explored nonlocal features of essentially bipartite systems. Multipartite nonlocality was first analysed in detail by Svetlichny \cite{svetlichny1} using a tripartite system. 
It was pointed out that the quantum mechanical violation of a $N$-partite ($N\ge 3$) BI is not sufficient for ensuring genuine $N$-partite nonlocality. This is because such $N$-partite BI violation can be simulated by nonlocal correlations between $M$ $(<N)$ particles, where $M$ nonlocally correlated particles can vary from one run of the experiment to the other. Hence, for certifying genuine multipartite nonlocality by considering the specific case of three-particle system, Svetlichny derived an inequality from the condition of locality satisfied across all its \emph{three} possible bipartitions. It is then guaranteed that the quantum violation of such an inequality cannot be reproduced in terms of any hybrid local-nonlocal system, in which the nonlocal correlations are present only between two particles, while these two particles are locally correlated with third particle. Subsequently, Svetlichny's approach was generalised for an arbitrary $N$-qubit system \cite{svetlichny2, collins}. Testability and interesting implications of Svetlichny's criterion of multipartite nonlocality were analysed by Mitchell \emph{et al} \cite{mitc}. 
 In this connection, it was also pointed out \cite{cerd} how the maximal quantum violation of the three particle case of Mermin's form of $N$-particle BI \cite{mermin} can be reproduced by a hybrid local-nonlocal model in which the nonlocal correlations are present only between two of the parties.  Thus, Mermin's inequality and its variants \cite{ardheli} cannot be used to probe genuine multipartite nonlocality. 
%

Against the above backdrop, our present paper explores a rigorous approach for analysing multipartite nonlocality by taking cue from the line of studies stimulated by Wigner's instructive demonstration \cite {wigner} of an incompatibility between quantum mechanics and local realism (LR) for the Einstein-Podolsky-Rosen-Bohm (EPRB) scenario involving a singlet state. For this purpose, Wigner's LR inequality was derived assuming, in conformity with the locality condition, the existence of joint probability distributions (JPDs) in the hidden variable space corresponding to the occurrence of different possible combinations of the outcomes of measurements of the relevant observables. However, Wigner's approach has not been much studied, apart from applying it for entangled neutral Kaons \cite{kaon} and its extensions for an arbitrary two qubit pure state \cite{castello} as well as for two qutrit systems \cite{wigner-ineq2}. Subsequently, a noteworthy work has been the generalisation of Wigner's approach to obtain suitable inequalities for probing nonlocality in the multipartite scenario \cite{wigner-ineq1}. Nevertheless, such generalised Wigner's inequality is unable to characterise nonlocality in a sense mentioned previously \textit{i.e.} the locality condition assumed while deriving such inequality can be violated by a $N$-partite state in which $M$ $(<N)$ particles are nonlocally correlated.\\~\\
%
%
%
%
In this paper we consider nonlocality in the following sense: If all the $(2^{N-1} -1)$ different Wigner's local realist (WLR) inequalities corresponding to all possible $(2^{N-1} -1)$ different bipartitions are individually violated by quantum mechanics, then such a state is rigorously certified to be nonlocal, meaning that any one of the subsystems is nonlocally correlated with all other subsystems. Thus, what is essentially required for such certification of nonlocality for any given state is that the individual quantum violations of all these inequalities need to be demonstrated separately for independent sets of experimental runs, irrespective of whether the measurement settings are the same or different for these separate sets of experimental runs. In our study, the measurement settings for which the individual quantum violations of our formulated WLR inequalities occur are, in general, different for different such inequalities. Now, in order to obtain these desired WLR inequalities, we proceed by assuming the existence of joint probability  distributions satisfying the locality condition across all possible distinct bipartitions of the multipartite system, where each bipartition consists of $n$ and $N-n$ particles respectively. A key feature of this scheme is that the violation of any such individual inequality would signify nonlocality for the specific bipartition corresponding to the inequality which is considered, irrespective of whether other inequalities of the set are individually violated.\\~\\

The plan of the paper is as follows. We begin by recapitulating the essence of Wigner's original derivation for a singlet state and its subsequent extension for an arbitrary multipartite state (Sec. \ref{wig_ineq}). Next, in order to explain our formulation of the further generalisation of Wigner's approach for detecting multipartite nonlocality, we proceed by first illustrating the basic idea of our scheme for tripartite states (Sec. \ref{tri_genuine}). This is followed by the formulation of a complete set of WLR inequalities whose individual quantum violations would show nonlocality for an arbitrary multipartite state across all of its different bipartitions (Sec. \ref{multi_genuine}). Then in (Sec.\ref{sec5}, the efficacy of our scheme for detecting nonlocality of different tripartite and quadripartite states, as compared to Svetlichny's inequality, is discussed. Furthermore, in (Sec. \ref{sec5}), we explain the way the individual quantum violation of any such inequality would enable identification of the corresponding bipartition for which nonlocality holds good. In conclusion (Sec. \ref{conclusion}), we summarize the salient features of our work and indicate a few directions for future study. 

\section{Background: Wigner's inequality and its generalisation in the multipartite scenario}\label{wig_ineq}

Let us consider the EPRB scenario involving two spin-$\frac{1}{2}$ particles prepared in a singlet state and shared by two spatially separated observers (Alice and Bob) who make measurements of the dichotomic observables $x_i^s$ where $i\in\{0,1\}$ and $s\in{A,B}$. The joint probabilities of the outcomes ($a=\pm ,b=\pm$) of such measurements are denoted by $P(x_k^A\pm,x_l^B\pm)$. The central ingredient of Wigner's derivation is the following: By considering an underlying stochastic hidden variable (HV) distribution corresponding to a quantum state, one assumes the existence of overall joint probabilities in the HV space for the individual outcomes of measurements of the relevant observables, from which the measured marginal probabilities are obtained by integration over the HV distribution. As a consequence, the locality condition is ensured because the single probability of the occurrence of an individual measurement outcome for any one of the particles (obtained as a marginal of the assumed overall joint probability distributions) is fixed, irrespective of what measurement is performed on the other particle. Thus, expressing the joint probabilities as marginals of the overall joint probability distributions and integrating over the relevant distribution in the HV space, one can obtain the original form of Wigner's inequality in terms of the observed joint probabilities:
\begin{eqnarray}\label{wig_original}
P(x_0^A+,x_0^B+)\le P(x_0^A+,x_1^B+)+P(x_1^A+,x_0^B+)\nonumber\\
\end{eqnarray}
which is quantum mechanically violated for a singlet state. Motivated by the above derivation, Wigner's inequality can immediately be obtained which is quantum mechanically violated for any two qubit entangled pure state. To this end, we proceed as follows: First, the overall joint probability pertaining to a particular hidden variable $\lambda$ is written as

\begin{eqnarray}\label{eq4}
P_\lambda(x_0^A+,x_0^B+)&=&\sum_{x_1^A,x_1^B}
P_\lambda(x_0^A+,
x_1^A;x_0^B+,
x_1^B)\nonumber\\
&=&P_\lambda(x_0^A+,
x_1^A+;x_0^B+,
x_1^B+)\nonumber\\&&+P_\lambda(x_0^A+,
x_1^A+;x_0^B+,
x_1^B-)\nonumber\\&&+P_\lambda(x_0^A+,
x_1^A-;x_0^B+,
x_1^B+)\nonumber\\&&+P_\lambda(x_0^A+,
x_1^A-;x_0^B+,
x_1^B-)\nonumber\\
\end{eqnarray}

Similar expressions follow for such $2^4$ joint probabilities corresponding to all measurement settings and outcomes considered in this scenario. Studying all these expressions one can obtain following relation such as    
\begin{eqnarray}
P_\lambda(x_0^A+,x_1^B+)+P_\lambda(x_1^A+,x_0^B+)\nonumber\\
+P_\lambda(x_1^A-,x_1^B-)
=P_\lambda(x_0^A+,x_0^B+)+\epsilon\nonumber\\
\end{eqnarray}
where $\epsilon\ge0$. 
After integrating over the HV space and assuming non-negativity of the overall JPDs, it immediately follows 
\begin{eqnarray}
P(x_0^A+,x_0^B+)-P(x_0^A+,x_1^B+)-\nonumber\\
P(x_1^A+,x_0^B+)
-P(x_1^A-,x_1^B-)\le 0
\end{eqnarray}
 whose maximal quantum mechanical violation occurs for a maximally entangled state.\\~\\
 
Proceeding similarly as discussed above, the generalised Wigner inequality (GWI) for the multipartite case was derived \cite{wigner-ineq1} which can be written in the following form:
\begin{eqnarray}\label{GWI_multi}
P(x_0^1+,x_0^2+,\dots ,x_{0}^N+)- P(x_1^1+,x_0^2+, \dots, x_{0}^N+)-\nonumber\\P(x_0^1+,x_1^2+, \dots, x_{0}^N+)-. . . -
 P(x_1^1-,x_1^2- ,\dots ,x_{1}^N-)\nonumber\\ \le0\nonumber\\
\end{eqnarray}
where measurement settings of the dichotomic observables deployed by the $s^{th}$ observer are $x_i^s$ where $s=1,\dots N$, $i\in\{0,1\}$ and the outcomes are $\{+,-\}$ respectively. For the special case of the tripartite system, we have 
\begin{eqnarray}\label{GWI}
P(x_0^1+,x_0^2+,x_0^3+)-P(x_1^1+,x_0^2+ ,x_0^3+)-\nonumber\\  P(x_0^1+,x_1^2+ ,x_0^3+)-P(x_0^1+,x_0^2+,x_1^3+)-\nonumber\\
 P(x_1^1-,x_1^2-,x_1^3-) \le0\nonumber\\
\end{eqnarray}
However, violation of the above inequality does not necessarily imply genuine nonlocality. Let us now consider the bi-separable state $\ket{0}\otimes\frac{1}{\sqrt2}(\ket{00}+\ket{11})$ where Alice has the state $\ket{0}$ while Bob and Charlie share the triplet state. It is then easily seen that the observable joint probabilities for suitable measurement settings violate the LR inequality given by Eq.(\ref{GWI}), although the correlation between measurements by Alice and Bob-Charlie  can be explained by a LR model. On the other hand, for nonlocality to be implied by the violation of a suitable form of GWI, it is necessary to ensure that all the subsystems are nonlocally correlated. Next, in order to derive such required form of GWI, we proceed as follows, by first considering the tripartite case, followed by suitable generalisation for an arbitrary $N$-partite state.
%

\section{Nonlocality using Wigner's approach for tripartite states}\label{tri_genuine}
Here for the tripartite system comprising the subsystems $A$, $B$ and $C$, we consider all possible bipartitions denoted by $A|BC$, $B|AC$ and $C|AB$ respectively. Two measurement settings per party, \emph{i.e.}, for $A$, $B$, $C$ respectively, are denoted by $x^A_i,x^B_i,x^C_i \hspace{.05 in}\text{where}\hspace{.05 in} i\in \{0,1\}$, and each measurement by an individual party is taken to yield two outcomes $a,b,c\in\{+,-\}$. First, applying the locality condition across the particular $A|BC$ cut, the existence of a joint probability distribution for a HV ($\lambda$) is assumed for such a bipartition so that 
the measured joint probability distribution is given by
\begin{equation}\label{eq:jpd}
P(abc|x^Ax^Bx^C)=\sum_{\lambda}q_{\lambda} \ P_{\lambda}(a|x^A)P_{\lambda}(bc|x^Bx^C), \ \sum_\lambda q_\lambda = 1,
\end{equation} 

 For instance, the observed probability of getting $(+,+,+)$ for $x^A_0,x^B_0,x^C_0$ is obtained as follows 
\be 
P(x^A_0+,x^B_0+,x^C_0+) =  \sum_\lambda q_\lambda \ P_\lambda(x^A_0+,x^B_0+,x^C_0+)
\ee 
where $\sum_\lambda q_\lambda = 1$ and
\be \label{eq:op1}
P_\lambda(x^A_0+,x^B_0+,x^C_0+) =  \sum P_\lambda (+,x^A_1,++,x^B_0x^C_1,x^B_1x^C_0,x^B_1x^C_1)
\ee
where the summation is taken over $x^A_1=\pm$ and $x^B_ix^C_j=++,+-,-+,--$ while the outcomes corresponding to the measurement settings $x_0^A$, $x_0^Bx_0^C$ are fixed given by $+$, $++$ respectively. Similarly, other observed joint probabilities can be obtained by using Eq.\eqref {eq:jpd}. \\

 In general, if the joint probability cannot be reproduced by any convex mixtures of HV models that are local across all possible bipartitions $A|BC, B|AC$ and $C|AC$, 
then the correlations is said to be genuinely nonlocal. However, here we consider a weaker version of the genuine nonlocality, in which the joint probability cannot be reproduced by local HV models across all bipartitions separately, that is, 
\begin{equation}\label{weakgn}
    P(abc|x^Ax^Bx^C) \neq 
    \begin{cases}
    \sum_{\lambda}q^1_{\lambda} \ P_{\lambda}(a|x^A)P_{\lambda}(bc|x^Bx^C) \\
    \sum_{\lambda}q^2_{\lambda} \ P_{\lambda}(b|x^B)P_{\lambda}(ac|x^Ax^C) \\
    \sum_{\lambda}q^3_{\lambda} \ P_{\lambda}(c|x^C)P_{\lambda}(ab|x^Ax^B),
    \end{cases}
\end{equation}
where $\sum_\lambda q^1_\lambda = \sum_\lambda q^2_\lambda = \sum_\lambda q^3_\lambda = 1$.
Now, we derive the theorem which provides the basis for our analysis of the notion of nonlocality according to \eqref{weakgn} for the tripartite case.

\begin{theorem}\label{thm1}
Nonlocality occurs if the following all three generalised WLR inequalities are quantum mechanically violated corresponding to the bipartitions $A|BC$, $B|AC$ and $C|AB$ respectively:
\begin{eqnarray}\label{I1-a-bc}
I_{A|BC} &=&
P(x^A_0+,x^B_0+,x^C_0+)-P(x^A_1+,x^B_0+,x^C_0+)\nonumber\\
&& -P(x^A_1-,x^B_1+,x^C_0+)-P(x^A_0+,x^B_1-,x^C_0+) \nonumber\\
&& -P(x^A_1-,x^B_1+,x^C_0-) - P(x^A_0+,x^B_1-,x^C_0-) \leqslant 0,\nonumber\\  
\\
I_{B|AC} &=&
P(x^A_0+,x^B_0+,x^C_0+)-P(x^A_0+,x^B_1+,x^C_0+)\nonumber\\
&& -P(x^A_1+,x^B_1-,x^C_0+)-P(x^A_1-,x^B_0+,x^C_0+) \nonumber\\
&& -P(x^A_1+,x^B_1-,x^C_0-) - P(x^A_1-,x^B_0+,x^C_0-) \leqslant 0 . \label{I1-b-ac}\nonumber \\
\text{and}\\
I_{C|AB} &=&
P(x^A_0+,x^B_0+,x^C_0+)-P(x^A_0+,x^B_0+,x^C_1+)\nonumber\\
&& -P(x^A_0+,x^B_1+,x^C_1-)-P(x^A_0+,x^B_1-,x^C_0+) \nonumber\\
&& -P(x^A_0-,x^B_1+,x^C_1-) - P(x^A_0-,x^B_1-,x^C_0+) \leqslant 0 . \label{I1-c-ab}\nonumber \\
\end{eqnarray}
\end{theorem}
The proof of this theorem is given in Appendix (\ref{appendix1}).
The individual violations of the inequalities Eqs. (\ref{I1-a-bc})-(\ref{I1-c-ab}) signify tripartite nonlocality rigorously. For example, consider the state $\ket{0}\otimes\frac{1}{\sqrt 2}(\ket{00}+\ket{11})$ for the bipartition $A|BC$ where one subsystem A is locally correlated with two other subsystems B and C who are nonlocally correlated among themselves. Such a state satisfies the inequality Eq.(\ref{I1-a-bc}), but may violate the other inequalities Eqs.(\ref{I1-b-ac}), (\ref{I1-c-ab}). Hence it is necessary that all the inequalities Eqs.(\ref{I1-a-bc})-(\ref{I1-c-ab}) are violated to ensure that any one subsystem is nonlocally correlated with all other subsystems of the tripartite system. \\~\\ 

Next, we briefly discuss a few illustrative usages of the generalised WLR inequalities. For this purpose, we consider the well known $\ket{GHZ}$ and $\ket{W}$ states written in the computational basis
\begin{eqnarray}
\ket{GHZ}=\frac{1}{\sqrt2}(\ket{000}+\ket{111})\label{3-ghz}\\
\ket{W}=\frac{1}{\sqrt3}(\ket{001}+\ket{010}+\ket{100})\label{3-w}
\end{eqnarray}

In order to obtain the maximum violations of Eqs. (\ref{I1-a-bc}), (\ref{I1-b-ac}), and (\ref{I1-c-ab}) for the above mentioned states, consider the following settings parametrised by $\alpha^j_i$ of the $j^{th}$ observer ($j$=1,2,3) measuring the observables $x_i$ (i=0,1)
\begin{eqnarray}\label{settings}
\ket{m^j_{x^j_i+}}=\cos\alpha^j_i\ket{0}+\sin\alpha^j_i\ket{1}\nonumber\\
\ket{m^j_{x^j_i-}}=-\sin\alpha^j_i\ket{0}+\cos\alpha^j_i\ket{1}
\end{eqnarray}
 
For the $W$-state, we obtain the maximum quantum violations of the inequalities Eqs.(\ref{I1-a-bc})-(\ref{I1-c-ab}) to be the same given by 0.101
in all the three distinct bipartitions. We specify the local measurement parameters for which the maximum violation of the LR inequality $I_{A|BC}$ is obtained are $\alpha_0^1\simeq 1.27$, $\alpha_1^1\simeq 0.29$, $\alpha_0^2\simeq 0$, $\alpha_1^2\simeq \frac{\pi}{4} $, $\alpha_0^3\simeq 0$ and $\alpha_1^3\simeq 1.45$ (in radian), see Table \ref{table1}. Similarly the corresponding values of the parameters for maximum violation of $I_{B|AC}$ and $I_{C|AB}$ are  $\alpha_0^1\simeq 0$, $\alpha_1^1\simeq \frac{\pi}{4}$, $\alpha_0^2\simeq 1.27$, $\alpha_1^2\simeq 0.29$, $\alpha_0^3\simeq 0$ and $\alpha_1^3\simeq 0.45$ and $\alpha_0^1\simeq \pi$, $\alpha_1^1\simeq 1.14$, $\alpha_0^2\simeq \pi$, $\alpha_1^2\simeq \frac{\pi}{4} $, $\alpha_0^3\simeq 1.27$ and $\alpha_1^3\simeq 0.29$ (in radian) respectively. Using the given measurement settings Eq.(\ref{settings}), the maximum quantum violation for GHZ state is found to be significantly smaller, given by $0.06$ for all the three distinct bipartitions. The parameters for which maximum violations occur for the bipartitions $I_{A|BC}$, $I_{B|AC}$ and $I_{C|AB}$, given in the same order as before, are $\big( 1.33, 1.80, 2.34, \frac{\pi}{2}, 2.35, .016\big)$, $\big(0.78, \frac{\pi}{2}, 1.80, 1.33, 2.35, 0.92\big)$ and $\big(2.35, 0.17, 2.35, \frac{\pi}{2}, 1.33, 1.80\big)$ (in radian) respectively, see Table \ref{table2}. Thus, we observe that maximum violations of the LR inequalities happen for different measurement settings for each of the inequalities. \\~\\

\begin{table}
\begin{center}
\begin{tabular}{ |p{0.8cm}|p{0.8cm}|p{0.8cm}|p{0.8cm}|p{0.8cm}|p{0.8cm}|p{0.8cm}|p{0.8cm} |}
\hline
\hline
$W$ & $\alpha_0^1$&$\alpha_1^1$&$\alpha_0^2$&$\alpha_1^2$&
$\alpha_0^3$&$\alpha_1^3$ & QV   \\ 
\hline 
$I_{A|BC}$ & 1.27&0.29&0&$\frac{\pi}{4}$&
0&1.45&0.101    \\ 
\hline 
$I_{B|AC}$& 0&$\frac{\pi}{4}$&1.27&0.29&
0&0.45&0.101  \\  
\hline 
$I_{C|AB}$ & $\pi$&1.14&$\pi$&$\frac{\pi}{4}$&1.27&0.29&0.101    \\ 
\hline
\hline
\end{tabular}
\caption{ Specification of the local measurement parameters
for which the maximum quantum violations (QV) of the generalised WLR inequalities Eqs. (\ref{I1-a-bc})-(\ref{I1-c-ab}) have been obtained for the W-state. }
\label{table1}
\end{center}
\end{table}

\begin{table}
\begin{center}
\begin{tabular}{ |p{0.9cm}|p{0.8cm}|p{0.8cm}|p{0.8cm}|p{0.8cm}|p{0.8cm}|p{0.8cm}|p{0.8cm} |}
\hline
\hline
$GHZ$ & $\alpha_0^1$&$\alpha_1^1$&$\alpha_0^2$&$\alpha_1^2$&
$\alpha_0^3$&$\alpha_1^3$&QV    \\ 
\hline 
$I_{A|BC}$ & 1.33&1.80&2.34 &$\frac{\pi}{2}$&2.35&0.016&0.06\\ 
\hline 
$I_{B|AC}$& $\frac{\pi}{4}$&$\frac{\pi}{2}$&1.80&1.33&2.35&0.92&0.06  \\  
\hline 
$I_{C|AB}$ & 2.35&0.17&2.35&$\frac{\pi}{2}$&1.33&1.80 &0.06 \\ 
\hline
\hline
\end{tabular}
\caption{ Specification of the local measurement parameters
for which the maximum quantum violations (QV) of the generalised WLR inequalities Eqs. (11)-(13) have been obtained for the GHZ-state. }
\label{table2}
\end{center}
\end{table}

Next, in order to investigate the tolerance to white noise of the optimal quantum violations of the above inequalities, we consider a tripartite mixed state given by
\begin{equation}
\rho=p|\psi\rangle\langle\psi|+\frac{1-p}{8}\mathbb{I}
\end{equation}

where $|\psi\rangle\langle\psi|$ corresponds to the tripartite pure state $\ket{\psi}$, $p$ is the visibility parameter, and $(1 - p)$
denotes the amount of white noise present in the state $\rho$. For $p=0$, $\rho$ denotes the maximally mixed state. The
minimum value of $p$ for which the mixed state $\rho$ violates a given local realist inequality is known as the threshold visibility pertaining to the state $\ket{\psi}$
corresponding to the considered inequality. In this case, such threshold visibilities for the $\ket{GHZ}$ and $\ket{W}$ states turn out to be $0.894$ and $0.832$ respectively. \\~\\
%
 
Next, in the subsequent section we will formulate the set of generalised WLR inequalities for detecting nonlocality for an arbitrary $N$-qubit state.

\section{Nonlocality using Wigner's approach for any multipartite state} \label{multi_genuine}
Generalising the scheme discussed in the preceding section, in this case we consider all possible bipartitions for obtaining the desired generalised WLR inequalities. For a given $N$-qubit state distributed amongst $N$ spatially separated observers, we need to construct $2^{N-1}-1$ 
\cite{Demianowicz} different LR inequalities whose individual quantum violations would certify nonlocality rigorously. \\~\\
To this end, we consider a typical bipartition of the $N$-partite system between $n$ and $N-n$ parties where $r_1, \dots, r_n$ denote the parties in one of the partitions and $s_1,\dots, s_{N-n}$ denote the parties in the other partition. For the $r_i^{th}$ party, the two measurement settings and the two outcomes per setting are denoted by $x^{r_i}_0,x^{r_i}_1$ and $a^{r_i}\in\{+,-\}$ respectively. Similarly, for the $s_i^{th}$ party in the other partition, the two measurement settings and the two outcomes per setting are denoted by $x^{s_i}_0,x^{s_i}_1$ and $a^{s_i}\in\{+,-\}$  .

%

Now, following Wigner's approach,
in conformity with the locality assumption being satisfied across the $n|N-n$ cut, we assume the existence of the joint probability distribution for a HV ($\lambda$) from which the observable joint probabilities can be obtained for the bipartition $r_1\cdots r_n|s_1\cdots s_{N-n}$. Thus, for instance, the joint probability of obtaining the outcomes $(+,+,\cdots,+)$ for the measurements of the observable $x_0$ performed on all $N$ subsystems respectively is given by
\begin{widetext}
\begin{equation}
P(x_0^{r_1}+,x_0^{r_2}+,\dots,x_0^{r_n}+,x_0^{s_1}+,x_0^{s_2}+,\dots,x_0^{s_{N-n}}+)=\sum_\lambda q_\lambda P_\lambda(x_0^{r_1}+,x_0^{r_2}+,\dots,x_0^{r_n}+,x_0^{s_1}+,x_0^{s_2}+,\dots,x_0^{s_{N-n}}+)
\end{equation}
where $\sum_\lambda q_\lambda=1$ and 
\begin{eqnarray}\label{jpd2}
&& P_\lambda(x_0^{r_1}+,x_0^{r_2}+,\dots,x_0^{r_n}+,x_0^{s_1}+,x_0^{s_2}+,\dots,x_0^{s_{N-n}}+)\nonumber\\ &=& \sum P_\lambda(++\cdots +, \ 
x_1^{r_1}x_0^{r_2}\cdots  x_0^{r_n},\ x_0^{r_1}x_1^{r_2}\cdots x_0^{r_n}, \dots, \ x_1^{r_1}x_1^{r_2}\cdots x_1^{r_n},++\dots +,\ x_1^{s_1}x_0^{s_2}\cdots x_0^{s_{N-n}},\dots,\  x_1^{s_1}x_1^{s_2}\cdots x_1^{s_{N-n}})\nonumber\\
\end{eqnarray}
The above summation over probability distributions in the HV space is taken over all possible combinations of outcomes of the relevant observables like $x_1^{r_1}x_0^{r_2}\cdots x_0^{r_n}$, $x_0^{r_1}x_1^{r_2}\cdots x_0^{r_n},\dots$, $x_1^{r_1}x_1^{r_2}\cdots x_1^{r_n}$ appearing in the different joint distributions occurring on the RHS of the above equation. The outcomes corresponding to the measurement settings $x_0^{r_1}\cdots x_0^{r_n}$, $x_0^{s_1}\cdots x_0^{s_{N-n}}$ are fixed, given by $+\cdots+$, $+\cdots+$ respectively.
\end{widetext}


Next, in order to detect nonlocality of a $N$-partite state by considering all possible $2^{N-1}-1$ bipartitions, we have to obtain the required $2^{N-1}-1$ inequalities by invoking the locality condition for each such bipartition. The way these inequalities can be formulated is illustrated by deriving a typical such generalised WLR inequality for the most general bipartition $r_1\cdots r_n|s_1\cdots s_{N-n}$. Subsequently, one can readily obtain the desired complete set of generalised WLR inequalities by putting $n=1,2,\dots,\lfloor\frac{N}{2}\rfloor$, where $\lfloor   \rfloor$ denotes the greatest integer function for a given value of $\frac{N}{2}$.   
Now, let us proceed to the following theorem:
\begin{widetext}
\begin{theorem}\label{thm2}
The following generalised WLR inequality, denoted by $I_{n|N-n}$, is derived assuming the locality condition across a typical bipartition $r_1\cdots r_n|s_1\cdots s_{N-n}$:
\bea
I_{n|N-n} &=& P(x_0^{r_1}+,x_0^{r_2}+, \dots, x_0^{r_n}+, x_0^{s_1}+,x_0^{s_2}+, \dots, x_0^{s_{N-n}}+) - P(x_1^{r_1}+,x_0^{r_2}+, \dots, x_0^{r_n}+, x_0^{s_1}+,x_0^{s_2}+, \dots, x_0^{s_{N-n}}+) \nonumber \\
&& - \sum_{\substack{x_1^{r_1},x_0^{r_2}, \dots, x_0^{r_n}\\ \neq (+,+,\dots,+)}} P(x_1^{r_1},x_0^{r_2}, \dots, x_0^{r_n}, x_1^{s_1}+,x_0^{s_2}+, \dots, x_0^{s_{N-n}}+) \nonumber  \\
&& - \sum_{\substack{x_1^{s_1},x_0^{s_2}, \dots, x_0^{s_{N-n}}\\ \neq (+,+,\dots,+)}} P(x_0^{r_1}+,x_0^{r_2}+, \dots, x_0^{r_n}+, x_1^{s_1},x_0^{s_2}, \dots, x_0^{s_{N-n}}) \le 0 
\label{eq:generalised-ineq}
\eea 
\end{theorem}
\end{widetext} 
 The proof of this theorem is given in Appendix \ref{appendix2}. Now, considering the particular cases of tripartite and quadripartite systems respectively, we will show explicitly how the complete set of generalised WLR inequalities for rigorously certifying nonlocality can be obtained from the generalised WLR inequality given by Eq.(\ref{eq:generalised-ineq}).

\subsubsection{Tripartite scenario}
In this case, first, for a specific bipartition $A|BC$, we obtain the following form of generalised WLR inequality satisfying the locality condition across this bipartition by putting $N=3$, $n=1$ in Eq.(\ref{eq:generalised-ineq}) 
\begin{eqnarray}\label{ineq-3-A-BC}
I_{A|BC} &=&
P(x_0^{r_1}+,x_0^{s_1}+,x_0^{s_2}+)-P(x_1^{r_1}+,x_0^{s_1}+,x_0^{s_2}+)\nonumber\\
&& -P(x_1^{r_1}-,x_1^{s_1}+,x_0^{s_2}+)-P(x_0^{r_1}+,x_1^{s_1}+,x_0^{s_2}-) \nonumber\\
&& -P(x_0^{r_1}+,x_1^{s_1}-,x_0^{s_2}+) - P(x_0^{r_1}+,x_1^{s_1}-,x_0^{s_2}-) \leqslant 0 . \nonumber \\
\end{eqnarray}

In the above inequality, by interchanging the measurement settings and the outcomes of measurements on the subsystems $r_1$ and $s_1$, one can find the following form of generalised WLR inequality for the bipartition $B|AC$

\begin{eqnarray}\label{ineq-3-B-AC}
I_{B|AC} &=&
P(x_0^{r_1}+,x_0^{s_1}+,x_0^{s_2}+)-P(x_0^{r_1}+,x_1^{r_1}+,x_0^{s_2}+)\nonumber\\
&& -P(x_1^{r_1}+,x_1^{s_1}-,x_0^{s_2}+)-P(x_1^{r_1}+,x_0^{s_1}+,x_0^{s_2}-) \nonumber\\
&& -P(x_1^{r_1}-,x_0^{s_1}+,x_0^{s_2}+) - P(x_1^{r_1}-,x_0^{s_1}+,x_0^{s_2}-) \leqslant 0  \nonumber\\ 
\end{eqnarray}
Similarly, the generalised WLR inequality for the bipartition $C|AB$ can be obtained from the inequality Eq.(\ref{ineq-3-A-BC}) by interchanging the measurement settings and the outcomes of measurements on the subsystems $r_1$ and $s_2$, given by
\begin{eqnarray}
\label{ineq-3-C-AB}
I_{C|AB} &=&
P(x_0^{r_1}+,x_0^{s_1}+,x_0^{s_2}+)-P(x_0^{r_1}+,x_0^{s_1}+,x_1^{s_2}+)\nonumber\\
&& -P(x_0^{r_1}+,x_1^{s_1}+,x_1^{s_2}-)-P(x_0^{r_1}-,x_1^{s_1}+,x_0^{s_2}+) \nonumber\\
&& -P(x_0^{r_1}+,x_1^{s_1}-,x_0^{s_2}+) - P(x_0^{r_1}-,x_1^{s_1}-,x_0^{s_2}+) \leqslant 0 . \nonumber \\
\end{eqnarray}

%
%
%
%

We recall that in Sec.(\ref{tri_genuine}) we had derived a set of LR inequalities Eqs.(\ref{I1-a-bc})-(\ref{I1-c-ab}) for detecting nonlocality in a tripartite system for all of its different bipartitions. It is to be noted that those inequalities involve combinations of JPDs which are significantly different from the combinations of JPDs occurring in the set of tripartite generalised WLR inequalities given by Eqs.(\ref{ineq-3-A-BC})-(\ref{ineq-3-C-AB}). Hence, these two sets of LR inequalities are not equivalent. We cannot obtain any one inequality belonging to either of the sets of inequalities Eqs.(\ref{I1-a-bc})-(\ref{I1-c-ab}) or Eqs.(\ref{ineq-3-A-BC})-(\ref{ineq-3-C-AB}) from any one inequality of the other set by relabelling the measurement settings and corresponding outcomes. In fact, a particular usefulness of the Wigner approach discussed here lies in providing a flexible framework for obtaining inequivalent sets of WLR inequalities for a given N-partite state, thereby providing an increased choice of the appropriate form of the inequality for detecting nonlocality of the multipartite states across its different bipartitions. We will illustrate this operational advantage in Sec. \ref{sec5}. \\~\\
We now compare the efficacy of the WLR inequalities Eqs.\ref{I1-a-bc})-(\ref{I1-c-ab}) with that of the WLR inequalities Eqs.(\ref{ineq-3-A-BC})-(\ref{ineq-3-C-AB}). As already mentioned in Sec.(\ref{tri_genuine}), the individual quantum violations of the former set of inequalities enable detection of nonlocality for the GHZ and W states across all different bipartitions. In contrast, the individual quantum violations of the latter set of inequalities occur for the W-state, but not for the GHZ state. Here note that the maximum quantum violations of the generalised WLR inequalities Eqs.(\ref{ineq-3-A-BC})-(\ref{ineq-3-C-AB}) for the W-state in all the three bipartitions are found to be the same given by $0.138$, which is slightly greater than the corresponding value $0.101$ for the individual quantum violations of the inequalities Eqs.(\ref{I1-a-bc})-(\ref{I1-c-ab}). We have used the measurement settings given by Eq.(\ref{settings}); the particular values of the relevant measurement parameters specified in Table \ref{table3} have been used to obtain the maximum quantum violations.\\~\\
\begin{table}
\begin{center}
\begin{tabular}{ |p{0.8cm}|p{0.8cm}|p{0.8cm}|p{0.8cm}|p{0.8cm}|p{0.8cm}|p{0.8cm}|p{0.8cm} |}
\hline
\hline
$W$ & $\alpha_0^1$&$\alpha_1^1$&$\alpha_0^2$&$\alpha_1^2$&
$\alpha_0^3$&$\alpha_1^3$  &QV  \\ 
\hline 
$I_{A|BC}$ & $\frac{\pi}{2}$ & $\frac{\pi}{4}$ &2.75&
0.39&0&1.55 &0.138   \\ 
\hline 
$I_{B|AC}$& 0.39&2.75&$\frac{\pi}{2}$&2.35&
$\pi$&0.06&0.138  \\  
\hline 
$I_{C|AB}$ & 0&0.22&2.75&0.39&$\frac{\pi}{2}$&0.78 &0.138   \\ 
\hline
\hline
\end{tabular}
\caption{ Specification of the local measurement parameters
for which the maximum quantum violations (QV) of the generalised WLR inequalities have been obtained for the W-state using Eqs.(\ref{ineq-3-A-BC})-(\ref{ineq-3-C-AB}). }
\label{table3}
\end{center}
\end{table}

\begin{table}
\begin{center}
\begin{tabular}{ |p{0.8cm}|p{0.8cm}|p{0.8cm}|p{0.8cm}|p{0.8cm}|p{0.8cm}|p{0.8cm}|p{0.8cm} |}
\hline
\hline
$\ket{\Psi}$ & $\alpha_0^1$&$\alpha_1^1$&$\alpha_0^2$&$\alpha_1^2$&
$\alpha_0^3$&$\alpha_1^3$ &QV   \\ 
\hline 
$I_{A|BC}$ & $\frac{\pi}{2}$ & $\frac{\pi}{4}$ &1.93&
1.20&0.45&2.40 & 0.15  \\ 
\hline 
$I_{B|AC}$ & 1.93&1.20&$\frac{\pi}{2}$&$\frac{\pi}{4}$&0.45&$\pi$ &0.15 \\  
\hline 
$I_{C|AB}$ &2&2.56&$\frac{\pi}{2}$&0.69&1.60&2.06&0.008    \\ 
\hline
\hline
\end{tabular}
\caption{ Specification of the local measurement parameters
for which the maximum quantum violations (QV) of the generalised WLR inequalities have been obtained for the state $\ket{\Psi}$ using Eqs.(\ref{ineq-3-A-BC})-(\ref{ineq-3-C-AB}). }
\label{table4}
\end{center}
\end{table}

 Next, consider the class of generalised W-state
\begin{equation}\label{genw3}
  \ket{W}_g=\cos\mu\ket{001}+\sin\mu\cos\theta\ket{010}
+\sin\mu\sin\theta\ket{100}  
\end{equation}

where $\mu\in[0,\pi]$ and $\theta\in[0,\pi]$.
A state belonging to a subclass of this state was shown to be useful for information processing tasks \cite{gen_w}.
Here, we consider a particular subclass of the state by substituting $\mu=\frac{\pi}{4}$
\begin{equation}\label{genw}
\ket{W^\prime}=\frac{1}{\sqrt2}(\ket{001}+\cos\theta\ket{010}
+\sin\theta\ket{100})
\end{equation}

\begin{theorem}

For the generalised-W state of the form Eq. (\ref{genw}), quantum mechanics violates the generalised WLR inequalities Eqs.(\ref{ineq-3-A-BC})-(\ref{ineq-3-C-AB}) upon suitable local measurements on each of the subsystems. The maximum quantum violations with respect to $\theta$ $\in[0,\pi]$, are found to be 
\begin{eqnarray}
I^{max}_{A|BC}=I^{max}_{B|AC}=\frac{1}{4}\Big(\sqrt{1+\sin^2 2\theta}-1\Big)\nonumber\\
I^{max}_{C|AB}=\frac{\cos^2\theta}{\Big(1+\cos^2\theta\Big)\Big(1+\sqrt{1+\frac{4\cos^2\theta}{(1+\cos^2\theta)^2}}\Big)} .
\end{eqnarray}

\end{theorem}
\begin{proof}
In order to achieve maximum violation, we have used the measurement settings given by Eq.(\ref{settings}). Thus, the quantities $I_{A|BC}$, $I_{B|AC}$ and $I_{C|AB}$ are functions of $six$ local parameters $\alpha^j_i$, where $i\in\{0,1\}$ and $j\in\{1,2,3\}$ and the state parameter $\theta$. In Table \ref{table} we have shown the parameter values for which the quantum violations of the three WLR inequalities for the given state are obtained. \\~\\

\begin{table}
\begin{center}
\begin{tabular}{ |p{3.2cm}||p{5.5cm}|  }
\hline
\hline
$I_{A|BC}$ & $\alpha_0^1=\frac{\pi}{2},\hspace{.05 in} \alpha_0^2+\alpha_1^2=\pi \hspace{.05 in}\text{and}\hspace{.05in} \alpha_0^3=0$     \\ 
\hline 
$I_{B|AC}$ & $\alpha_0^1+\alpha_1^1=\pi,\hspace{.05 in} \alpha_0^2=\frac{\pi}{2} \hspace{.05 in}\text{and}\hspace{.05in} \alpha_0^3=0$    \\ 
\hline 
$I_{C|AB}$ & $\alpha_0^1=\frac{\pi}{2},\hspace{.05 in} \alpha_0^2+\alpha_1^2=\pi \hspace{.05 in}\text{and}\hspace{.05in} \alpha_0^3=0$  
 \\ 
\hline 
\hline
\end{tabular}
\caption{ Specification of the local measurement parameters
for which quantum violations of the LR inequalities Eqs.(\ref{ineq-3-A-BC})-(\ref{ineq-3-C-AB}) have been computed. }
\label{table}
\end{center}
\end{table}

After substituting the above mentioned parameter values in the expressions of $I_{A|BC}$, $I_{B|AC}$ and $I_{C|AB}$ respectively, we obtain the following simplified expressions:
\begin{eqnarray}
I_{A|BC}&=&-\frac{1}{4}\Big(2\sin^2\alpha_0^2+\sin2
\alpha_0^2\sin2\alpha_1^1\sin2\theta\Big)\label{maxi1}\\
I_{B|AC}&=&-\frac{1}{4}\Big(2\sin^2\alpha_1^1+\sin2
\alpha_1^1\sin2\alpha_1^2\sin2\theta\Big)\\
I_{C|AB}&=&\frac{1}{4}\Big(-(3+\cos2\theta)\sin^2\alpha_1^2+2\sin2
\alpha_1^2\sin2\alpha_1^3\cos\theta\Big)\nonumber\\
\end{eqnarray}
Let us consider the specific case for $I_{A|BC}$. To maximize $I_{A|BC}$ with respect to the local parameters, we have to solve following equations for $\alpha_0^2$ and $\alpha_1^1$
\begin{eqnarray}
\frac{\partial I_{A|BC}}{\partial \alpha_1^1}&=&0 \hspace{.1 in} {\text{and}}\\
\frac{\partial I_{A|BC}}{\partial \alpha_0^2}&=&0
\end{eqnarray}
It is easy to verify that the conditions are satisfied for $\alpha_1^1=\frac{\pi}{4}$ and $\alpha_0^2=\frac{1}{2}
\tan^{-1}(-\sin2\theta)$ respectively. Now putting these values into Eq.(\ref{maxi1}), we obtain $I^{max}_{A|BC}$. In a similar fashion $I^{max}_{B|AC}$ and $I^{max}_{C|AB}$ can be achieved.

\end{proof}

\begin{figure}[t]
\includegraphics[width=8cm]{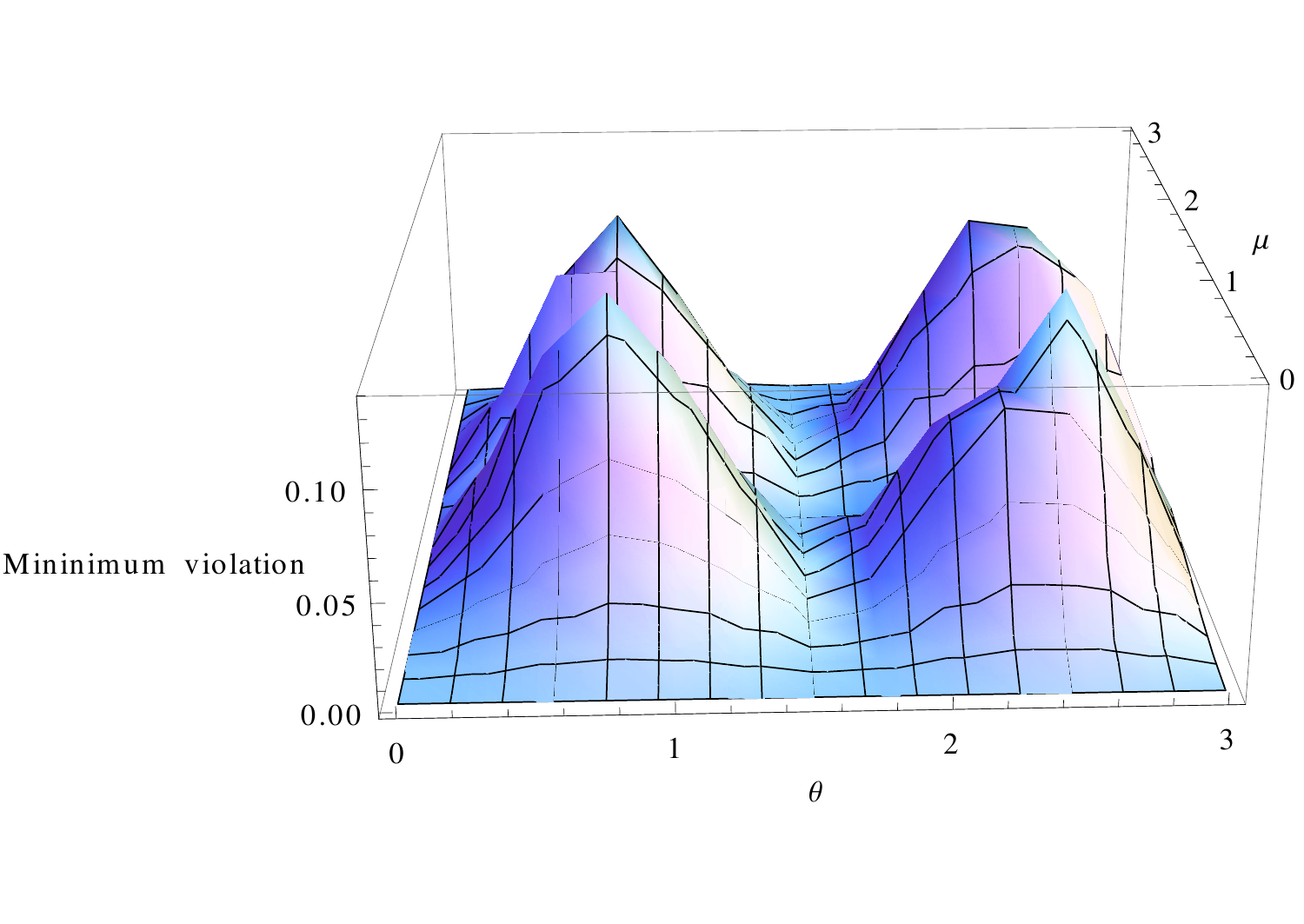}
\centering
\caption{Minima of the quantum violations of the generalised WLR inequalities Eqs.(\ref{ineq-3-A-BC})-(\ref{ineq-3-C-AB}) for the generalised $W$-state given by Eq.(\ref{genw3}) are plotted numerically ($\theta$ and $\mu$ are in radian). The plot shows that no quantum violation can be found for $\theta=\frac{\pi}{2}$ and $\mu=\frac{\pi}{2}$. }
\end{figure}\label{plot1}

The plot in Fig.(\ref{plot1}) shows numerically obtained minima of the quantum violations of the generalised WLR inequalities Eqs.(\ref{ineq-3-A-BC})-(\ref{ineq-3-C-AB}) as functions of the state parameters $\theta$ and $\mu$ of the generalised W state. 
\\~\\

\subsubsection{Quadripartite scenario}
Let us consider a quadripartite system comprising four spatially separated subsystems $A$, $B$, $C$ and $D$. In this case,  bipartitions can occur in different ways. For example, any one of the subsystems can be locally /nonlocally correlated with three other subsystems, \emph{i.e.,} these bipartitions are denoted as $A|BCD$, $B|ACD$, $C|ABD$, and $D|ABC$. The other type of bipartitions involving two subsystems in each group are denoted as $AB|CD$, $AC|BD$ and $AD|BC$. Thus, these seven different bipartitions give rise to seven different generalised WLR inequalities.  \\~\\
Let us first consider the way the generalised WLR inequalities $I_{1|234}$, $I_{2|134}$, $I_{3|124}$ and $I_{4|123}$ are obtained for the bipartitions $A|BCD$, $B|ACD$, $C|ABD$, and $D|ABC$ respectively. 
For instance, for the specific bipartition $A|BCD$, the following generalised WLR inequality satisfying the locality condition across this bipartition is derived by putting $N=4$, $n=1$ in the generalised WLR Eq.(\ref{eq:generalised-ineq})

\begin{widetext}
\begin{eqnarray}\label{ineq-4-A-BCD}
I_{1|234} &=&
P(x_0^{r_1}+,x_0^{s_1}+,x_0^{s_2}+,x_0^{s_3}+)-P(x_1^{r_1}+,x_0^{s_1}+,x_0^{s_2}+,x_0^{s_3}+)\nonumber\\
&& - P(x_1^{r_1}-,x_1^{s_1}+,x_0^{s_2}+,x_0^{s_3}+)-
P(x_0^{r_1}+,x_1^{s_1}+,x_0^{s_2}+,x_0^{s_3}-)\nonumber\\
&& - P(x_0^{r_1}+,x_1^{s_1}+,x_0^{s_2}-,x_0^{s_3}+)- P(x_0^{r_1}+,x_1^{s_1}-,x_0^{s_2}+,x_0^{s_3}+) \nonumber\\
&& - P(x_0^{r_1}+,x_1^{s_1}+,x_0^{s_2}-,x_0^{s_3}-)-
P(x_0^{r_1}+,x_1^{s_1}-,x_0^{s_2}-,x_0^{s_3}+) \nonumber\\
&& - P(x_0^{r_1}+,x_1^{s_1}-,x_0^{s_2}+,x_0^{s_3}-)-
P(x_0^{r_1}+,x_1^{s_1}-,x_0^{s_2}-,x_0^{s_3}-)\nonumber\\
&&\leqslant 0 . \nonumber \\
\end{eqnarray}
\end{widetext}

Similar to the previously considered tripartite case, the desired inequalities $I_{2|134}$, $I_{3|124}$ and $I_{4|123}$ for the other bipartitions $B|ACD$, $C|ABD$ and $D|ABC$ respectively can be obtained by interchanging the measurement settings and outcomes of measurements used in the above inequality Eq.(\ref{ineq-4-A-BCD}).
%

Next, one can readily obtain the required inequalities $I_{12|34}$, $I_{13|24}$ and $I_{14|23}$ for the bipartitions $AB|CD$, $AC|BD$, and $AD|BC$ respectively. For example, $I_{12|34}$ follows from Eq.({\ref{eq:generalised-ineq}) by putting $N=4$, $n=2$, given by
\begin{widetext}
    \begin{eqnarray}\label{ineq-4-AB-CD}
I_{12|34} &=&
P(x_0^{r_1}+,x_0^{r_2}+,x_0^{s_1}+,x_0^{s_2}+)-P(x_1^{r_1}+,x_0^{r_2}+,x_0^{s_1}+,x_0^{s_2}+)\nonumber\\
&& - P(x_1^{r_1}+,x_0^{r_2}-,x_1^{s_1}+,x_0^{s_2}+)-
P(x_1^{r_1}-,x_0^{r_2}+,x_1^{s_1}+,x_0^{s_2}+) \nonumber\\
&& - P(x_1^{r_1}-,x_0^{r_2}-,x_1^{s_1}+,x_0^{s_2}+)- P(x_0^{r_1}+,x_0^{r_2}+,x_1^{s_1}+,x_0^{s_2}-) \nonumber\\
&& - P(x_0^{r_1}+,x_0^{r_2}+,x_1^{s_1}-,x_0^{s_2}+)-
P(x_0^{r_1}+,x_0^{r_2}+,x_1^{s_1}-,x_0^{s_2}-)\nonumber\\
&&\leqslant 0  \nonumber \\
\end{eqnarray}
\end{widetext}

\blk
Similarly, one can obtain the other inequalities $I_{12|34}$, $I_{13|24}$ and $I_{14|23}$ by interchanging the measurement settings and the outcomes of measurements used in the above inequality Eq.(\ref{ineq-4-AB-CD}). In this way,  the entire set consisting of seven generalised WLR inequalities $\{I_{1|234}, I_{2|134}, I_{3|124}, I_{4|123}, I_{12|34}, I_{13|24}, I_{14|23}\}$ can be derived. Individual violations of all such seven inequalities would imply nonlocality of a given quadripartite state across all of its different bipartitions.\\~\\
 Now, let us mention a few salient features of the set of WLR inequalities obtained for different bipartitions of a quadripartite state by considering some specific examples.   
It is to be noted that we will use the measurement settings given by Eq.(\ref{settings}) to obtain the quantum violations in each cases. The generalised $GHZ$ state for a quadripartite system given by $\ket{GHZ}_g=\cos{\theta}\ket{0000}+
 \sin{\theta}\ket{1111}$ is found to violate all the generalised WLR inequalities for $\theta=1.45$ \text{radian}. The $W$-state for a quadripartite system given by $\ket{W}=\frac{1}{2}(\ket{0001}+\ket{0010}+\ket{0100}+\ket{1000})$ violates all the seven generalised WLR inequalities. We illustrate in Fig. (\ref{plot2}) quantum mechanical violation of generalised W state given by 

 \begin{eqnarray}
  \ket{W}_g&=&\cos\theta\ket{0001}+\sin\theta \sin\mu\ket{0010}+\nonumber\\
   && \sin\theta \cos\mu \sin\nu\ket{0100}+\sin\theta \cos\mu \cos\nu\ket{1000}\nonumber\\
 \end{eqnarray}

  where $\theta\in\{0,\pi\}$, $\mu\in\{0,\pi\}$ and $\nu=\frac{\pi}{4}$.
 
\begin{figure}[h]
\includegraphics[width =.4\textwidth]{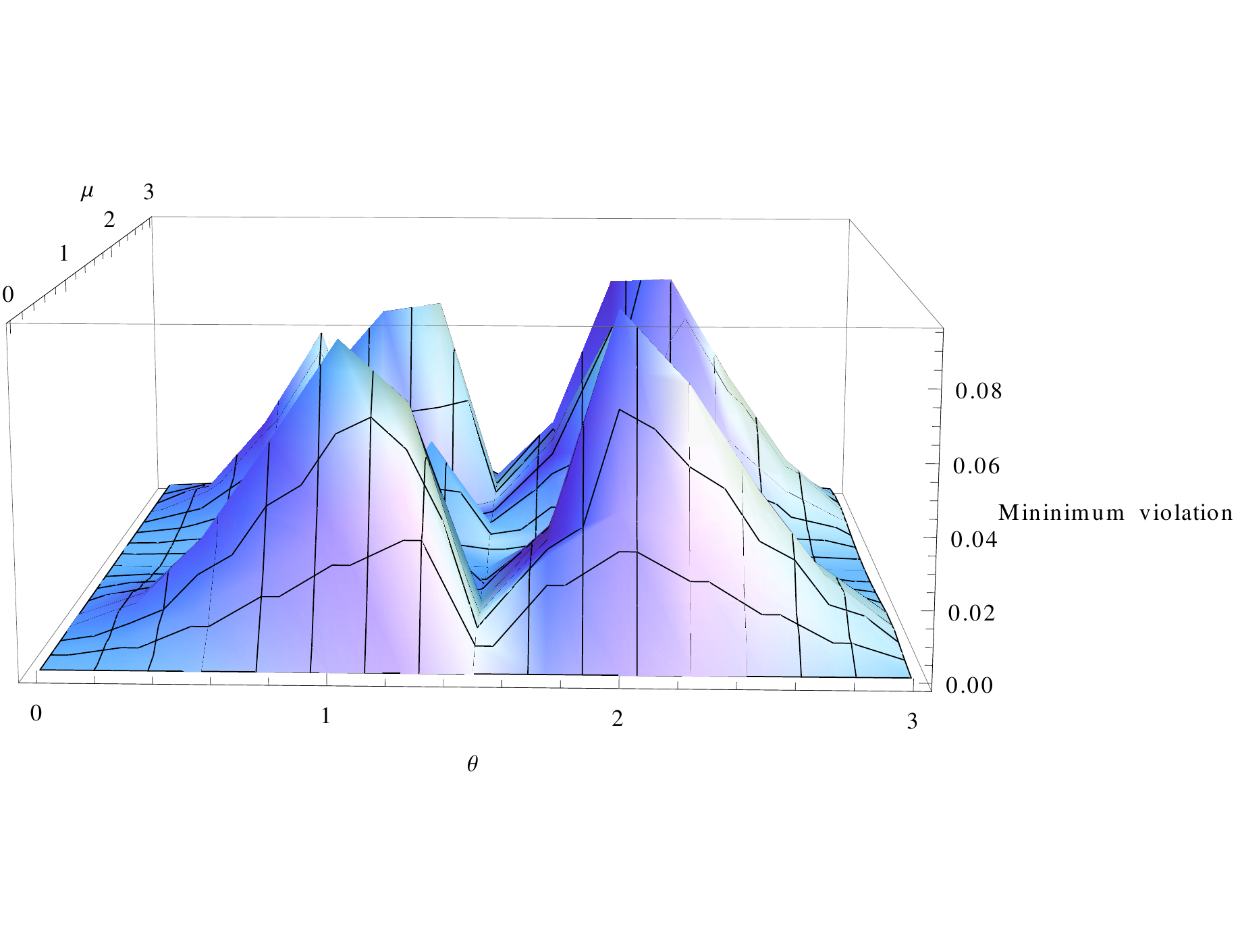}
\caption{Minima of the quantum violations of the seven generalised WLR inequalities given by Eqs.(\ref{ineq-4-A-BCD}), (\ref{ineq-4-AB-CD}) and other similar inequalities obtained for all the different bipartitions are plotted as functions of the two state parameters $\theta$ and $\mu$ measured in radian respectively, by taking the fixed value of the other state parameter $\nu=\frac{\pi}{4}$.}
\end{figure}\label{plot2}
 \section{Efficacy of the generalised WLR inequalities}\label{sec5}

 In this section, some significant consequences of our formulated generalised WLR inequalities are discussed. First, considering the generalised tripartite W state given by Eq.(\ref{genw3}), note that the minima of the quantum mechanical violations of the generalised WLR inequalities Eqs.(\ref{ineq-3-A-BC})-(\ref{ineq-3-C-AB}) have been plotted in Fig.(\ref{plot1}) as functions of the different values of the two state parameters. Here an important point to be noted that for the values of the parameters $\mu=\frac{\pi}{4}$ and $\theta \in\{2.0,3.0\}$, we have not obtained any violation of Svetlichny's inequality after numerical optimization studies for such states. Thus, this exemplifies that there exists certain states whose nonlocality cannot be detected using Svetlichny's inequality, but can be detected using the WLR inequalities given by Eqs.(\ref{ineq-3-A-BC})-(\ref{ineq-3-C-AB}). 
 For further illustration of the efficacy of these inequalities compared to Svetlichny's inequality, let us consider the following tripartite state \cite{bancal}
 
 \begin{equation}\label{ms}
 \ket{\Psi}=
\frac{\sqrt 3}{2}\ket{000}+\frac{\sqrt 3}{4}\ket{110}+\frac{1}{4}\ket{111}    
 \end{equation}

  It has been shown \cite{bancal} that the above state $\ket{\Psi}$ does not violate Svetlichny's inequality, whereas we have found that the generalised WLR inequalities Eqs.(\ref{ineq-3-A-BC})-(\ref{ineq-3-C-AB}) are all quantum mechanically violated for the state $\ket{\Psi}$, with the maximum quantum violations being $0.15$, $0.15$ and $0.008$ for the bipartitions $A|BC$, $B|AC$ and $C|AB$ respectively, see Table \ref{table4}.\\~\\
  Further, we have found that no quantum mechanical violation can be obtained for the generalisation of Svetlichny's inequality \cite{svetlichny2} for the quadripartite $W$-state denoted by $\ket{\mathcal{W}}$ where
\begin{equation}
\ket{\mathcal{W}}= \frac{1}{2}(\ket{0001}+\ket{0010}+\ket{0100}+\ket{1000}) 
\end{equation}
  
On the other hand, our generalised seven WLR inequalities are all individually violated for the $\ket{\mathcal{W}}$ state, thereby certifying its nonlocality. \\~\\
Let us also consider another example of a quadripartite state given by
\begin{equation}\label{1_quadripartite_state}
\ket{\Phi}=\frac{1}{\sqrt{26}}(\ket{0000}+\ket{+000}+
\ket{-+++}+\ket{0111})
\end{equation} 

where the Hadamard basis states $\ket{\pm}$ occurring in the second and third terms are given by $\ket{\pm}=\frac{1}{\sqrt2}(\ket{0}\pm\ket{1})$. Note that, while no quantum violation of the generalised Svetlichny inequality \cite{svetlichny2} can be obtained for the state given by Eq.(\ref{1_quadripartite_state}), the generalised seven quadripartite WLR inequalities given by Eqs.(\ref{ineq-4-A-BCD}), (\ref{ineq-4-AB-CD}) and other similar inequalities obtained for the rest of the bipartitions are all individually violated for this state. It is thus evident that, similar to the tripartite case, there are quadripartite states whose nonlocality cannot be detected by Svetlichny's inequality, but is detectable using the generalised WLR inequalities. \\~\\


Finally, we discuss another significant advantage of our formulated scheme based on appropriate multipartite generalisation of Wigner's approach. Let us consider a $N$-partite state which is not genuinely nonlocal, but may be nonlocal for a specific bipartition with respect to the $m|(N-m)$ cut. Now, in order to investigate the possibility of detecting this nonlocality, the generalised WLR inequality for this particular bipartition is useful whose quantum violation would signify nonlocality in this particular subspace. In order to illustrate this, let us consider the product quadripartite state $\ket{W}\otimes\ket{0}$ ($\ket{W}=\frac{1}{\sqrt3}(\ket{001}+\ket{010}+\ket{100}$) for which the generalised Svetlichny inequality is not violated. But, interestingly, nonlocality of this state for the specific bipartitions $AB|CD$, $AC|BD$ and $AD|BC$ can be detected through quantum violations of the respective generalised WLR inequalities applied to these bipartitions (the maximum quantum violations are found to be $0.44$, $0.44$ and $0.36$ respectively). However, the above considered state does not violate all the seven generalised quadripartite  WLR inequalities given by Eqs.(\ref{ineq-4-A-BCD}), (\ref{ineq-4-AB-CD}) and other similar inequalities for all the different bipartitions. It is thus interesting that although the given state cannot be regarded as genuinely nonlocal, the nonlocality existing in certain specific bipartitions can be detected using our generalised Wigner's approach - a feature that may be useful in applications based on multipartite nonlocality. 
\section{Conclusion}\label{conclusion} 
Our present paper, generalising the approach originally suggested by Wigner, serves to validate the following feature: In a multipartite system, the assumed existence of overall joint probabilities of the relevant observables in a HV theory which yield the marginal probabilities satisfying the locality condition across all different bipartitions is sufficient to demonstrate multipartite nonlocality. A distinctive feature of our generalisation of Wigner's approach is that it enables detecting whether nonlocality is present in a given bipartition of a multipartite state. While illustrative examples of this aspect have been provided in this paper using a suitable form of tripartite and quadripartite states, the possibility of more such examples needs to be investigated. This feature should be useful for experimentally probing finer characterisations of multipartite nonlocality in different subspaces of the entangled state and can provide more flexibility in harnessing multipartite nonlocality for potential applications. Such a line of investigation would, thus, complement the considerable studies that have been made concerning entanglement detection in different subspaces of a multipartite entangled state \cite{ent_sub1,ent_sub2,ent_sub3}.\\~\\

Here we note that, apart from the generalised Svetlichny inequality \cite{svetlichny2} and our generalised Wigner approach, another scheme \cite{bancal2} for studying multipartite nonlocality has been based on the  multipartite multidimensional generalisation of the CGLMP inequality for two qutrit states. The generalised multipartite multidimensional LR inequality thus obtained reduces to Svetlichny's inequality for the triqubit states. Thus, as a follow up to our present paper, a comprehensive comparative study between various aspects of our formulated generalised WLR inequalities and the above mentioned  multipartite multidimensional LR inequality should be worthwhile. Another notable line of studies \cite{bancal4, bancal5, dsaha, caval2} is based on the notion of nonlocality different from that defined by Svetlichny and followed in our work. Along this direction, a number of multipartite LR inequalities have been proposed in \cite{bancal4, dsaha}. In particular, the tripartite case has been extensively studied in \cite{bancal5, caval2}; for example, 185 facet LR inequalities have been formulated by considering the no-signaling bilocal polytope \cite{bancal5}. It should be an interesting direction of future research to comprehensively compare the implications of these inequalities with that of the inequalities obtained by our present approach. In this connection, it may also be mentioned that a class of $46$ LR inequalities had been derived to probe nonlocality in the tripartite case \cite{sliwa, sliwa1}. However, the applications of all these inequalities have been essentially restricted to the GHZ type states. 
%



Next, considering the applicational aspect, we note that the nonlocality stemming from the multipartite states has already been recognised as a potential resource for performing information processing protocols. For example, the resource theoretic aspects of multipartite nonlocality have been analysed in \cite{gisin3}. A possible application in devising a quantum key distribution protocol involving multipartite states, commonly known as conference key agreement, has been pointed out \cite{mike}, while its fully device independent treatment has also been developed \cite{wehner} using Mermin's multipartite 
 LR inequality. However, as we had mentioned in (Sec. \ref{intro}), the violation of Mermin's inequality does not certify genuine nonlocality. Therefore, for such applications, it would be interesting to probe the usefulness of our formulated generalised WLR inequalities. Here it is relevant to note an earlier study \cite{bruss} showing that the efficacy of the conference key agreement protocol can be enhanced by extracting higher key rate using multipartite entanglement. Thus,  in this context, the role of multipartite nonlocality should be instructive to analyze by using the generalised WLR inequalities. \\~\\


Finally, we would like to mention a possible future line of research concerning the quantitative relationship between entanglement and nonlocality in a multipartite scenario. While a number of studies have revealed interesting facets of the quantitative incommensurability between entanglement and nonlocality for a bipartite system \cite{brun,acin2,jphys,global2}, this issue has remained largely unexplored for a multipartite system. It should therefore be worthwhile to thoroughly investigate the commensurability between the maximum violations of the generalised WLR inequalities and the various measures of genuine multipartite entanglement such as the triangle measure \cite{triangle}, and the global measure of entanglement of a given multipartite state \cite{global1}. 


\section*{Acknowledgement}
SN acknowledges support from the Department of Science and Technology, Government of India
through the QuEST grant DST/ICPS/QuEST/2018/98. DS acknowledges support from National Post Doctoral Fellowship (PDF/2020/001682). DH acknowledges support from the NASI Senior
Scientist fellowship. ASM acknowledges support from the project no. DST/ICPS/QuEST/2018/98 of the Department of Science and Technology, Government of India. We also thank Siddhartha Das for his comments on the manuscript.

\bigskip 

\onecolumngrid

\appendix 

\section{Proof of Theorem \ref{thm1}}\label{appendix1}
Here we proceed to provide the proof of Theorem (1) by first considering the following relevant joint probability distributions in the hidden variable (HV) space:

\begin{eqnarray}\label{a2}
 P_\lambda(x^A_1+,x^B_0+,x^C_0+)&=&\sum_{x^B_0x^C_1,
x^B_1x^C_0,x^B_1x^C_1} P_\lambda (x_0^A+,x_1^A+,x^B_0x^C_0=++,x^B_0x^C_1,x^B_1x^C_0,x^B_1x^C_1)  \nonumber  \\
&& + \underbrace{\sum_{x^B_0x^C_1,
x^B_1x^C_0,x^B_1x^C_1} P_\lambda (x_0^A-,x_1^A+,x^B_0x^C_0=++,x^B_0x^C_1,x^B_1x^C_0,x^B_1x^C_1)}_{\delta_1}\nonumber\\
&=&\sum_{x^B_0x^C_1,
x^B_1x^C_0,x^B_1x^C_1} P_\lambda (x_0^A+,x_1^A+,x^B_0x^C_0=++,x^B_0x^C_1,x^B_1x^C_0,x^B_1x^C_1)+ \ \delta_1\nonumber\\
\end{eqnarray}
where $\delta_1$ is sum of the JPDs, and hence, $\delta_1\ge 0$. Similarly we write,

\begin{eqnarray}
&& P_\lambda(x^A_1-,x^B_1+,x^C_0+)  \nonumber  \\
 &=&\sum_{x^B_0x^C_0,
x^B_0x^C_1,x^B_1x^C_1} P_\lambda (x_0^A+,x_1^A-,x^B_0x^C_0,x^B_0x^C_1,x^B_1x^C_0=++,x^B_1x^C_1) +\nonumber\\&& \sum_{x^B_0x^C_0,
x^B_0x^C_1,x^B_1x^C_1} P_\lambda (x_0^A-,x_1^A-,x^B_0x^C_0,x^B_0x^C_1,x^B_1x^C_0=++,x^B_1x^C_1)\nonumber\\
&=& \sum_{
x^B_0x^C_1,x^B_1x^C_1} P_\lambda (x_0^A+,x_1^A-,x^B_0x^C_0=++,x^B_0x^C_1,x^B_1x^C_0=++
,x^B_1x^C_1) \nonumber\\ 
&&+ \underbrace{\sum_{x_0^Bx_0^C\neq\{++\},
x^B_0x^C_1,x^B_1x^C_1} P_\lambda (+,-,x_0^Bx_0^C,x^B_0x^C_1,++,x^B_1x^C_1)  +\sum_{x^B_0x^C_0,
x^B_0x^C_1,x^B_1x^C_1} P_\lambda (-,-,x^B_0x^C_0,x^B_0x^C_1,++,x^B_1x^C_1)}_{\delta_2}\nonumber\\
&=& \sum_{
x^B_0x^C_1,x^B_1x^C_1} P_\lambda (x_0^A+,x_1^A-,x^B_0x^C_0=++,x^B_0x^C_1,
x^B_1x^C_0=++,x^B_1x^C_1)+\ \delta_2\nonumber\\
\end{eqnarray}
where we have used the shorthand notation $P_\lambda (+,-,x_0^Bx_0^C,x^B_0x^C_1,++,x^B_1x^C_1)$ to denote 
$P_\lambda (x_0^A+,x_1^A-,x_0^Bx_0^C,x^B_0x^C_1,x_1^Bx_0^C=++
,x^B_1x^C_1)$ and so on. 

\begin{eqnarray}
&& P_\lambda(x^A_0+,x^B_1-,x^C_0+) \nonumber  \\
&=& \sum_{x^B_0x^C_0,
x^B_0x^C_1,x^B_1x^C_1} P_\lambda (x_0^A+,x_1^A+,x^B_0x^C_0,x^B_0x^C_1,
x_1^Bx_0^C=-+,x^B_1x^C_1)\nonumber\\&&
+ \sum_{x^B_0x^C_0,
x^B_0x^C_1,x^B_1x^C_1} P_\lambda (x_0^A+,x_1^A-,x^B_0x^C_0,x^B_0x^C_1,
x_1^Bx_0^C=-+,x^B_1x^C_1)\nonumber\\
&=& \sum_{x_0^Bx_0^C,x^B_0x^C_1,x^B_1x^C_1} P_\lambda (x_0^A+,x_1^A+,x^B_0x^C_0,x^B_0x^C_1,
x_1^Bx_0^C=-+,x^B_1x^C_1) \nonumber\\ 
&& + \underbrace{\sum_{x^B_0x^C_1,x^B_1x^C_1}P_\lambda (+,-,++,x^B_0x^C_1,-+,x^B_1x^C_1) +
\sum_{
x_0^Bx_0^C\neq\{++\},x^B_0x^C_1,x^B_1x^C_1} P_\lambda (+,-,x_0^Bx_0^C,x^B_0x^C_1,-+,x^B_1x^C_1)}_{\delta_3} \nonumber\\
&=& \sum_{x^B_0x^C_1,x^B_1x^C_1}P_\lambda (x_0^A+,x_1^A-,x_0^Bx_0^C=++,x^B_0x^C_1,
x_1^Bx_0^C=-+,x^B_1x^C_1)+\ \delta_3
\end{eqnarray}
\begin{eqnarray}
&& P_\lambda(x^A_1-,x^B_1+,x^C_0-) \nonumber  \\
&=& \sum_{x^B_0x^C_0,
x^B_0x^C_1,x^B_1x^C_1} P_\lambda (x_0^A+,x_1^A-,x^B_0x^C_0,x^B_0x^C_1,
x_1^Bx_0^C=+-,x^B_1x^C_1) \nonumber\\
&&+ \sum_{x^B_0x^C_0,
x^B_0x^C_1,x^B_1x^C_1} P_\lambda (x_0^A-,x_1^A-,x^B_0x^C_0,x^B_0x^C_1,
x_1^Bx_0^C=+-,x^B_1x^C_1)\nonumber\\
&=& \sum_{
x^B_0x^C_1,x^B_1x^C_1} P_\lambda (x_0^A+,x_1^A-,x_0^Bx_0^C=++,x^B_0x^C_1,
x_1^Bx_0^C=+-,x^B_1x^C_1) \nonumber\\
&& + \underbrace{\sum_{x_0^Bx_0^C\neq\{++\},
x^B_0x^C_1,x^B_1x^C_1} P_\lambda (+,-,x_0^Bx_0^C,x^B_0x^C_1,+-,x^B_1x^C_1)  + \sum_{x^B_0x^C_0,
x^B_0x^C_1,x^B_1x^C_1} P_\lambda (-,-,x^B_0x^C_0,x^B_0x^C_1,+-,x^B_1x^C_1)}_{\delta_4} \nonumber\\
&=& \sum_{
x^B_0x^C_1,x^B_1x^C_1} P_\lambda (x_0^A+,x_1^A-,x_0^Bx_0^c=++,x^B_0x^C_1,
x_1^Bx_0^C=+-,x^B_1x^C_1)+\ \delta_4
\end{eqnarray}
\begin{eqnarray}\label{a6}
&& P_\lambda(x^A_0+,x^B_1-,x^C_0-) \nonumber  \\
&=&\sum_{x^B_0x^C_0,
x^B_0x^C_1,x^B_1x^C_1} P_\lambda (x_0^A+,x_1^A+,x^B_0x^C_0,x^B_0x^C_1,
x_1^Bx_0^C=--,x^B_1x^C_1)\nonumber\\&& 
+ \sum_{x^B_0x^C_0,
x^B_0x^C_1,x^B_1x^C_1} P_\lambda (x_0^A+,x_1^A-,x^B_0x^C_0,x^B_0x^C_1,
x_1^Bx_0^C=--,x^B_1x^C_1)\nonumber\\
&=& \sum_{
x^B_0x^C_1,x^B_1x^C_1} P_\lambda (x_0^A+,x_1^A-,x_0^Bx_0^C=++,x^B_0x^C_1,
x_1^Bx_0^C=--,x^B_1x^C_1) \nonumber  \\
&& + \underbrace{\sum_{x^B_0x^C_0,
x^B_0x^C_1,x^B_1x^C_1} P_\lambda (+,+,x^B_0x^C_0,x^B_0x^C_1,--,x^B_1x^C_1) +\sum_{x_0^Bx_0^C\neq\{++\},
x^B_0x^C_1,x^B_1x^C_1} P_\lambda (+,-,x_0^Bx_0^C,x^B_0x^C_1,--,x^B_1x^C_1)}_{\delta_5} \nonumber\\ 
&=& \sum_{
x^B_0x^C_1,x^B_1x^C_1} P_\lambda (x_0^A+,x_1^A-,x_0^Bx_0^C=++,x^B_0x^C_1,
x_1^Bx_0^C=--,x^B_1x^C_1)+\ \delta_5
\end{eqnarray}
In the above expressions $\delta_i$s are sums of JPDs, and hence, positive. Finally, we write
\begin{eqnarray}\label{jpd4} 
 P_\lambda(x^A_0+,x^B_0+,x^C_0+)=\sum_{x^B_0x^C_1,
x^B_1x^C_0,x^B_1x^C_1} P_\lambda (x_0^A+,x_1^A+,x^B_0x^C_0=++,x^B_0x^C_1,x^B_1x^C_0,
x^B_1x^C_1) \nonumber  \\
+ \sum_{x^B_0x^C_1,x^B_1x^C_1} P_\lambda (x_0^A+,x_1^A-,x^B_0x^C_0=++,x^B_0x^C_1,x^B_1x^C_0=++,
x^B_1x^C_1) \nonumber  \\
+ \sum_{x^B_0x^C_1,x^B_1x^C_1} P_\lambda (x_0^A+,x_1^A-,x^B_0x^C_0=++,x^B_0x^C_1,x^B_1x^C_0=-+,
x^B_1x^C_1) \nonumber  \\
+ \sum_{x^B_0x^C_1,x^B_1x^C_1} P_\lambda (x_0^A+,x_1^A-,x^B_0x^C_0=++,x^B_0x^C_1,x^B_1x^C_0=+-,
x^B_1x^C_1) \nonumber  \\
+ \sum_{x^B_0x^C_1,x^B_1x^C_1} P_\lambda (x_0^A+,x_1^A-,x^B_0x^C_0=++,x^B_0x^C_1,x^B_1x^C_0=--,
x^B_1x^C_1) \nonumber\\ 
\end{eqnarray}

 Now adding Eqs. (\ref{a2})-(\ref{a6}) and further using the expansion of $P_\lambda(x^A_0+,x^B_0+,x^C_0+)$ in \eqref{jpd4}, we obtain
\begin{eqnarray}
&& P_\lambda(x^A_1+,x^B_0+,x^C_0+) +P_\lambda(x^A_1-,x^B_1+,x^C_0+)+P_\lambda(x^A_0+,x^B_1-,x^C_0+) +P_\lambda(x^A_1-,x^B_1+,x^C_0-) + P_\lambda(x^A_0+,x^B_1-,x^C_0-)\nonumber\\
&=& \sum_{x^B_0x^C_1,
x^B_1x^C_0,x^B_1x^C_1} P_\lambda (+,+,++,x^B_0x^C_1,x^B_1x^C_0,x^B_1x^C_1) 
+ \sum_{x^B_0x^C_1,x^B_1x^C_1} P_\lambda (+,-,++,x^B_0x^C_1,++,x^B_1x^C_1) \nonumber  \\
&& + \sum_{x^B_0x^C_1,x^B_1x^C_1} P_\lambda (+,-,++,x^B_0x^C_1,-+,x^B_1x^C_1) 
+ \sum_{x^B_0x^C_1,x^B_1x^C_1} P_\lambda (+,-,++,x^B_0x^C_1,+-,x^B_1x^C_1) \nonumber  \\
&& + \sum_{x^B_0x^C_1,x^B_1x^C_1} P_\lambda (+,-,++,x^B_0x^C_1,--,x^B_1x^C_1)+\sum_{l=1}^5\delta_l\nonumber\\
&=& P_\lambda(x^A_0+,x^B_0+,x^C_0+) +\sum_l\delta_l .
\end{eqnarray}

Now $\sum_{l=1}^5\delta_l$ is necessarily a positive quantity. Invoking this fact, we finally obtain,
\begin{eqnarray}
P_\lambda(x^A_1+,x^B_0+,x^C_0+) +P_\lambda(x^A_1-,x^B_1+,x^C_0+)+P_\lambda(x^A_0+,x^B_1-,x^C_0+) +P_\lambda(x^A_1-,x^B_1+,x^C_0-) + P_\lambda(x^A_0+,x^B_1-,x^C_0-)\nonumber\\
\ge P_\lambda(x^A_0+,x^B_0+,x^C_0+) \nonumber\\
\end{eqnarray}  
We rewrite it as 
\begin{eqnarray}\label{eqa9}
P_\lambda(x^A_0+,x^B_0+,x^C_0+)-P_\lambda(x^A_1+,x^B_0+,x^C_0+) -P_\lambda(x^A_1-,x^B_1+,x^C_0+)-P_\lambda(x^A_0+,x^B_1-,x^C_0+) -P_\lambda(x^A_1-,x^B_1+,x^C_0-) -\nonumber\\ P_\lambda(x^A_0+,x^B_1-,x^C_0-)
\le0\nonumber\\
\end{eqnarray}
Subsequently, integrating Eq.(\ref{eqa9}) over the HV distribution, we readily obtain the LR inequality given by Eq.(\ref{I1-a-bc}) in the main text. One can also obtain the LR inequalities for the bipartitions $B|AC$ and $C|AB$ Eq.(\ref{I1-b-ac}) and Eq.(\ref{I1-c-ab}) by simply
interchanging the measurement settings and the outcomes of measurements on the subsystems $A$, $B$ and $A$, $C$ respectively, occurring in the inequality given by Eq.(\ref{I1-a-bc}).

\section{Proof of Theorem \ref{thm2}}\label{appendix2}
We can write the JPD $ P_{\lambda}(x_0^{r_1}+,x_0^{r_2}+, \dots, x_0^{r_n}+, x_0^{s_1}+,x_0^{s_2}+, \dots, x_0^{s_{N-n}}+)$ as follows:
\bea\label{jpd5}
&&  P_{\lambda}(x_0^{r_1}+,x_0^{r_2}+, \dots, x_0^{r_n}+, x_0^{s_1}+,x_0^{s_2}+, \dots, x_0^{s_{N-n}}+) \nonumber \\
&=& \sum \ P_\lambda(++\cdots +,  x_1^{r_1}x_0^{r_2}\cdots x_0^{r_n}, \dots, x_1^{r_1}x_1^{r_2}\cdots x_1^{r_n}, ++\cdots +, x_1^{s_1}x_0^{s_2}\cdots x_0^{s_{N-n}}, \dots, x_1^{s_1}x_1^{s_2}\cdots x_1^{s_{N-n}}) \nonumber \\
\end{eqnarray}
where the summation is taken over all possible combinations of outcomes of the relevant observables appearing in the different joint probability distributions occurring on the RHS of the above equation. For convenience, we introduce the following short hand notation 
\be 
\Vec{x}_k = x_{k_1}^{r_1}x_{k_2}^{r_2}\cdots x_{k_n}^{r_n} , \quad \Vec{y}_t = x_{t_1}^{s_1}x_{t_2}^{s_2}\cdots x_{t_n}^{s_{N-n}} 
\ee 
where $k$ is the decimal form of binary string $(k_n\cdots k_2k_1)$ and $t$ is the decimal form of binary string $(t_{N-n}\cdots t_2t_1)$. Thus, $k$ takes value from 0 to $2^n-1$, and $t$ takes value from 0 to $2^{N-n}-1$. For instance, $\Vec{x}_0 = x_{0}^{r_1}x_{0}^{r_2}\cdots x_{0}^{r_n}$, $\Vec{y}_1 = x_{1}^{s_1}x_{0}^{s_2}\cdots x_{0}^{s_{N-n}}$. We further denote $|r|= 2^n-1$, and $|s| = 2^{N-n}-1$.
Expanding the following JPDs using this notation, we obtain

\bea \label{GIt2}
&& P_\lambda(x_1^{r_1}+,x_0^{r_2}+, \dots, x_0^{r_n}+, x_0^{s_1}+,x_0^{s_2}+, \dots, x_0^{s_{N-n}}+)  \nonumber \\
&=& \sum_{\Vec{y}_1, \dots, \Vec{y}_{|s|}} \sum_{\Vec{x}_0, \Vec{x}_2,\dots, \Vec{x}_{|r|}} \ P_\lambda \left(\Vec{x}_0 , \  \Vec{x}_1= ++\cdots +, \ \Vec{x}_2 , \ \dots, \ \Vec{x}_{|r|},\ \Vec{y}_0 = ++\cdots +, \ \Vec{y}_1, \ \dots, \ \Vec{y}_{|s|} \right) \nonumber \\ 
&=& \sum_{\Vec{y}_1, \dots, \Vec{y}_{|s|}} \sum_{\Vec{x}_2,\dots, \Vec{x}_{|r|}} \ P_\lambda \left(\Vec{x}_0 = ++\cdots +,\ \Vec{x}_1 = ++\cdots +,\ \Vec{x}_2,\ \dots, \ \Vec{x}_{|r|},\ \Vec{y}_0 = ++\cdots +, \ \Vec{y}_1, \ \dots, \ \Vec{y}_{|s|} \right) + \nonumber\\
&&\sum_{\Vec{y}_1,\Vec{y}_2, \dots, \Vec{y}_{|s|}} \sum_{\substack{\Vec{x}_0,\Vec{x}_2\dots, \Vec{x}_{|r|} \\ \Vec{x}_0 \neq (++\cdots+)}} \ P_\lambda \left(\Vec{x}_0,\ \Vec{x}_1 = ++\cdots +,\ \Vec{x}_2,\ \dots, \ \Vec{x}_{|r|},\ \Vec{y}_0 = ++\cdots +, \ \Vec{y}_1, \ \dots, \ \Vec{y}_{|s|} \right) 
 \nonumber \\
&=&\sum_{\Vec{y}_1, \dots, \Vec{y}_{|s|}} \sum_{\Vec{x}_2,\dots, \Vec{x}_{|r|}} \ P_\lambda \left(\Vec{x}_0 = ++\cdots +,\ \Vec{x}_1 = ++\cdots +,\ \Vec{x}_2,\ \dots, \ \Vec{x}_{|r|},\ \Vec{y}_0 = ++\cdots +, \ \Vec{y}_1, \ \dots, \ \Vec{y}_{|s|} \right)+\ \tilde{\delta}_1
\eea 

\bea \label{GIt3}
&& \sum_{\substack{x_1^{r_1},x_0^{r_2}, \dots, x_0^{r_n}\\ \neq (+,+,\dots,+)}} P_\lambda (x_1^{r_1},x_0^{r_2}, \dots, x_0^{r_n}, x_1^{s_1}+,x_0^{s_2}+, \dots, x_0^{s_{N-n}}+) \nonumber \\
&=& \sum_{\Vec{y}_0, \Vec{y}_2, \dots, \Vec{y}_{|s|}} \sum_{\substack{\Vec{x}_0,\Vec{x}_1,\dots, \Vec{x}_{|r|} \\ \Vec{x}_1 \neq (++\cdots+) }} P_\lambda \left(\Vec{x}_0 ,\ \Vec{x}_1,\ \Vec{x}_2,\ \dots, \ \Vec{x}_{|r|},\ \Vec{y}_0, \ \Vec{y}_1 = ++\cdots+, \ \Vec{y}_2 , \ \dots, \ \Vec{y}_{|s|} \right)  \nonumber \\
&=& \sum_{\Vec{y}_2, \dots, \Vec{y}_{|s|}} \sum_{\substack{\Vec{x}_1,\dots, \Vec{x}_{|r|} \\ \Vec{x}_1 \neq (+\cdots+) }} P_\lambda \left(\Vec{x}_0 = +\cdots +,\ \Vec{x}_1,\ \Vec{x}_2,\ \dots, \ \Vec{x}_{|r|},\ \Vec{y}_0 = +\cdots +, \ \Vec{y}_1 = +\cdots+, \ \Vec{y}_2 , \ \dots, \ \Vec{y}_{|s|} \right) + \ \tilde{\delta}_2 \nonumber \\
\eea 

\bea \label{GIt4}
&& \sum_{\substack{x_1^{s_1},x_0^{s_2}, \dots, x_0^{s_{N-n}}\\ \neq (+,+,\dots,+)}} P_\lambda (x_0^{r_1}+,x_0^{r_2}+, \dots, x_0^{r_n}+, x_1^{s_1},x_0^{s_2}, \dots, x_0^{s_{N-n}})  \nonumber \\ 
&=& \sum_{\substack{\Vec{y}_0,\Vec{y}_1, \dots, \Vec{y}_{|s|} \\ \Vec{y}_1 \neq (++\cdots+)} } \ \sum_{\Vec{x}_1,\dots, \Vec{x}_{|r|} } P_\lambda \left(\Vec{x}_0 = ++\cdots +,\ \Vec{x}_1,\ \Vec{x}_2,\ \dots, \ \Vec{x}_{|r|},\ \Vec{y}_0 = ++\cdots +, \ \Vec{y}_1, \ \Vec{y}_2 , \ \dots, \ \Vec{y}_{|s|} \right) \nonumber \\
&=& \sum_{\substack{\Vec{y}_1, \dots, \Vec{y}_{|s|} \\ \Vec{y}_1 \neq (++\cdots+)} } \ \sum_{\substack{\Vec{x}_1,\dots, \Vec{x}_{|r|} \\ \Vec{x}_1 \neq (++\cdots+) }} P_\lambda \left(\Vec{x}_0 = ++\cdots +,\ \Vec{x}_1,\ \Vec{x}_2,\ \dots, \ \Vec{x}_{|r|},\ \Vec{y}_0 = ++\cdots +, \ \Vec{y}_1, \ \Vec{y}_2 , \ \dots, \ \Vec{y}_{|s|} \right) + \ \tilde{\delta}_3 \nonumber \\
\eea 
In the above expressions $\tilde{\delta}_i$s are the sums of joint probability distributions, and hence, positive. Finally, we obtain 
\begin{eqnarray}\label{GIt1}
&& P_{\lambda}(x_0^{r_1}+,x_0^{r_2}+, \dots, x_0^{r_n}+, x_0^{s_1}+,x_0^{s_2}+, \dots, x_0^{s_{N-n}}+) \nonumber \\
&=& \sum_{\Vec{y}_1, \dots, \Vec{y}_{|s|}} \sum_{\Vec{x}_1,\dots, \Vec{x}_{|r|}} \ P_\lambda \left(\Vec{x}_0 = ++\cdots +, \  \Vec{x}_1, \ \dots, \ \Vec{x}_{|r|},\ \Vec{y}_0 = ++\cdots +, \ \Vec{y}_1, \ \dots, \ \Vec{y}_{|s|} \right) \nonumber \\
& = & \sum_{\Vec{y}_1, \dots, \Vec{y}_{|s|}} \sum_{\Vec{x}_2,\dots, \Vec{x}_{|r|}} \ P_\lambda \left(\Vec{x}_0 = ++\cdots +,\ \Vec{x}_1 = ++\cdots +,\ \Vec{x}_2,\ \dots, \ \Vec{x}_{|r|},\ \Vec{y}_0 = ++\cdots +, \ \Vec{y}_1, \ \dots, \ \Vec{y}_{|s|} \right) \nonumber \\
 && + \sum_{\Vec{y}_2, \dots, \Vec{y}_{|s|}} \sum_{\substack{\Vec{x}_1,\dots, \Vec{x}_{|r|} \\ \Vec{x}_1 \neq (++\cdots+) }} P_\lambda \left(\Vec{x}_0 = ++\cdots +,\ \Vec{x}_1,\ \Vec{x}_2,\ \dots, \ \Vec{x}_{|r|},\ \Vec{y}_0 = ++\cdots +, \ \Vec{y}_1 = ++\cdots+, \ \Vec{y}_2 , \ \dots, \ \Vec{y}_{|s|} \right) \nonumber \\ 
  && + \sum_{\substack{\Vec{y}_1, \dots, \Vec{y}_{|s|} \\ \Vec{y}_1 \neq (++\cdots+)} } \ \sum_{\substack{\Vec{x}_1,\dots, \Vec{x}_{|r|} \\ \Vec{x}_1 \neq (++\cdots+) }} P_\lambda \left(\Vec{x}_0 = ++\cdots +,\ \Vec{x}_1,\ \Vec{x}_2,\ \dots, \ \Vec{x}_{|r|},\ \Vec{y}_0 = ++\cdots +, \ \Vec{y}_1, \ \Vec{y}_2 , \ \dots, \ \Vec{y}_{|s|} \right) , 
\eea
where the summation $\Vec{x}_k,\Vec{y}_t$ is taken over all possible outcomes. In the second step of the above equation, we expand the summation over $\Vec{x}_1,\Vec{y}_1$ into three different terms. 
Now, adding Eqs.\eqref{GIt2}-\eqref{GIt4} and using the expansion of $P_\lambda(x_0^{r_1}+,x_0^{r_2}+, \dots, x_0^{r_n}+, x_0^{s_1}+,x_0^{s_2}+, \dots, x_0^{s_{N-n}}+)$ given in \eqref{GIt1}, we find that
\bea
&& P_\lambda(x_1^{r_1}+,x_0^{r_2}+, \dots, x_0^{r_n}+, x_0^{s_1}+,x_0^{s_2}+, \dots, x_0^{s_{N-n}}+) \nonumber \\
&& + \sum_{\substack{x_1^{r_1},x_0^{r_2} \dots, x_0^{r_n}\\ \neq (+,\dots,+)}} P_\lambda(x_1^{r_1},x_0^{r_2}, \dots, x_0^{r_n}, x_1^{s_1}+,x_0^{s_2}+, \dots, x_0^{s_{N-n}}+)  + \sum_{\substack{x_1^{s_1},x_0^{s_2} \dots, x_0^{s_{N-n}}\\ \neq (+,\dots,+)}} P_\lambda(x_0^{r_1}+,x_0^{r_2}+, \dots, x_0^{r_n}+, x_1^{s_1},x_0^{s_2}, \dots, x_0^{s_{N-n}}) \nonumber \\
&=& P_\lambda(x_0^{r_1}+,x_0^{r_2}+, \dots, x_0^{r_n}+, x_0^{s_1}+,x_0^{s_2}+, \dots, x_0^{s_{N-n}}+)+ \tilde{\delta}_1 + \tilde{\delta}_2 + \tilde{\delta}_3
\eea 
As like before, $\tilde{\delta}_1 + \tilde{\delta}_2 + \tilde{\delta}_3 \ge0$. Thus we obtain,
\begin{eqnarray}
&& P_\lambda(x_0^{r_1}+,x_0^{r_2}+, \dots, x_0^{r_n}+, x_0^{s_1}+,x_0^{s_2}+, \dots, x_0^{s_{N-n}}+)- P_\lambda(x_1^{r_1}+,x_0^{r_2}+, \dots, x_0^{r_n}+, x_0^{s_1}+,x_0^{s_2}+, \dots, x_0^{s_{N-n}}+) \nonumber \\
&& - \sum_{\substack{x_1^{r_1},x_0^{r_2} \dots, x_0^{r_n}\\ \neq (+,\dots,+)}} P_\lambda(x_1^{r_1},x_0^{r_2}, \dots, x_0^{r_n}, x_1^{s_1}+,x_0^{s_2}+, \dots, x_0^{s_{N-n}}+) \nonumber \\
&& - \sum_{\substack{x_1^{s_1},x_0^{s_2} \dots, x_0^{s_{N-n}}\\ \neq (+,\dots,+)}} P_\lambda(x_0^{r_1}+,x_0^{r_2}+, \dots, x_0^{r_n}+, x_1^{s_1},x_0^{s_2}, \dots, x_0^{s_{N-n}})\le0\nonumber\\
\end{eqnarray}
Finally, integrating the above equation over the HV distribution, we obtain the inequality given by Eq.(\ref{eq:generalised-ineq}).

\end{document}